\documentclass[AMA,STIX1COL]{WileyNJD-v2}

\articletype{}%

\usepackage{tikz}

\received{}
\revised{}
\accepted{}

\begin{document}

\title{The Impact of Irrational Behaviours in the Optional Prisoner's Dilemma with Game-Environment Feedback} 

\author[1]{Leonardo Stella*}

\author[2,3]{Dario Bauso}


\authormark{Stella and Bauso}

\address[1]{\orgdiv{Department of Computing, College of Science \& Engineering}, \orgname{University of Derby}, \orgaddress{\country{United Kingdom}}}

\address[2]{\orgdiv{Jan C. Willems Center for Systems and Control, ENTEG, Faculty of Science and Engineering}, \orgname{University of Groningen}, \orgaddress{\country{The Netherlands}}}

\address[3]{\orgdiv{Dipartimento di Ingegneria}, \orgname{University of Palermo}, \orgaddress{\country{Italy}}}


\corres{*Corresponding author, \email{l.stella@derby.ac.uk}}


\abstract[Summary]{In the optional prisoner's dilemma (OPD), players can choose to cooperate and defect as usual, but can also abstain as a third possible strategy. This strategy models the players' participation in the game and is a relevant aspect in many settings, e.g. social networks or opinion dynamics where abstention is an option during an election. In this paper, we provide a formulation of the OPD where we consider irrational behaviours in the population inspired by prospect theory. Prospect theory has gained increasing popularity in recent times thanks to its ability to capture aspects such as reference dependence or loss aversion which are common in human behaviour. This element is original in our formulation of the game and is incorporated in our framework through pairwise comparison dynamics. Recently, the impact of the environment has been studied in the form of feedback on the population dynamics. Another element of novelty in our work is the extension of the game-environment feedback to the OPD in two forms of dynamics, the replicator and the pairwise comparison. The contribution of this paper is threefold. First, we propose a modelling framework where prospect theory is used to capture irrational behaviours in an evolutionary game with game-environment feedback. Second, we carry out the stability analysis of the system equilibria and discuss the oscillating behaviours arising from the game-environment feedback. Finally, we extend our previous results to the OPD and we discuss the main differences between the model resulting from the replicator dynamics and the one resulting from the pairwise comparison dynamics.}



\maketitle


\section{Introduction}\label{sec:intro}
Evolutionary game theory studies the evolution of strategic interactions in a population of rational decision makers based on suitable incentives, which are given in the form of strategy-dependent payoffs \cite{Smith_1973, Smith_1982}. The popularity of a given strategy is influenced by the system dynamics and the incentives given to an individual for choosing that strategy, thus leading to myopic behaviours where the favoured option is not the global best for the population. This is the case for the well-known prisoner's dilemma (PD), in which the dominant strategy is defection \cite{Axelrod_1981}. In the setting of finite games, many evolutionary dynamics have been used, including best response dynamics, the Brown-von Neuman-Nash dynamics and replicator dynamics to mention a few. The latter is recognised to be the one most widely used \cite{Hofbauer_1998}. However, regardless of the mathematical model used, the common game setting involves fixed incentives associated with each strategy. This limitation does not consider the role of the environment in affecting the game dynamics and therefore the frequency of the strategies in the population. 

A recent line of research, initiated in the work by Weitz \emph{et al.} \cite{Weitz_2016}, has attracted a great attention as it proposes a new way to model the reciprocal dependence of the payoffs on the environment and vice versa, taking into account the enhancement and degradation effects on the environmental resource. This is done through a direct dependence of the payoff matrix on the environment, giving rise to interesting system dynamics characterised by oscillatory behaviours and equilibria on the boundary of the phase space. The resulting oscillations correspond to closed periodic orbits. Limit cycles have been observed also when a time-scale difference between the game and the environment dynamics is considered \cite{Tilman_2020}.

This feedback mechanism is believed to be able to explain the complexity of real-world systems for which one often observes bi-directional feedback between the state of the environment and the incentives associated to each strategy \cite{Tilman_2020}. Interactions of this kind can be found in a wide range of disciplines, including sociology, economics and animal behaviour \cite{PaisNu_2013, Mullon_2017, Lee_2019, Estrela_2019}. An example of this in human systems is the individuals' decision to vaccinate or not, which would result in otherwise-preventable infectious diseases in children and then an increasing incentive to vaccinate \cite{Bauch_2003, Bauch_2004, Galvani_2007}. In ecology, an example can be found in plant nutrient acquisition where the competitive balance between the interactions of different species and the availability of environmental resources alter the nature of the competitive process of nitrogen fixation \cite{Menge_2008}. Another example is in the context of collective decision-making in honeybees where an initial attempt to incorporate game-environment feedback within the parameters of the model has been proposed \cite{Baar_2020}, giving rise to oscillations. 

The work by Su \emph{et al.} has taken another direction by studying the impact of the environment and its evolution through game transitions, namely a framework where players move from playing one game and its associated payoff matrix to another one based on the environmental resource \cite{Su_2019}. Chen and Szolnoki have applied the concept of game-environment feedback to punishment and inspection in the setting of governing the commons \cite{Chen_2018}. Finally, another line of research has explored the asymmetric aspects of this topic, both in the case of asymmetric evolutionary dynamics and in the case of asymmetrical game-environment feedback \cite{Hauert_2019, Shao_2019}.

Game-environment feedback, or in general coevolutionary game theory, has generated increasing interest in the control community, where the design of adequate control policies play a crucial role in achieving a specific system behaviour \cite{Morimoto_2016, Wang_2020}. An example of the use of game-environment feedback in control is the work by Paarporn \emph{et al.} where the controlled system exhibits oscillatory behaviours due to deprecation of the environmental resource \cite{Paarporn_2018}. Examples from such a wide range of research areas show that an understanding of the dynamics in the presence of game-environment feedback is a crucial aspect for a deeper awareness of complex real-life systems.

Although several approaches have been adopted to study the prisoner's dilemma, including, e.g., the design of a strategy that is obtained by saturating a polynomial function in an $n$-population game setting \cite{Giordano_2018}, incorporating the evolution of the environment in the strategy profile through the payoff matrix poses a risk in limiting the likelihood of cooperation, and gives rise to an ``oscillating tragedy of the commons'' \cite{Weitz_2016}. An additional difficulty may come when we take into consideration the corresponding optional game, i.e. the optional prisoner's dilemma (OPD), where players can choose to cooperate and defect, but also abstain from playing the game \cite{Batali_1995}. This formulation includes the dynamics of many real-life decision-making settings, such as opinion dynamics in the context of an election campaign and duopolistic competition in marketing \cite{Batali_1995, Stella_2019}. 

By extending the work on game-environment feedback to the OPD, we derive two models: one resulting from the replicator dynamics as an extension of the prior work on game-environment feedback and another one resulting from the pairwise comparison dynamics. The motivation that underpins the need for this additional model comes from the surge of interest that prospect theory has had in the past years and its continuous use in the past 30 years within behavioural economics and beyond \cite{Barberis_2013}. Prospect theory has its foundations on the work from two psychologists, Kahneman and Tversky, through which they demonstrated a systematic misalignment between the predictions conducted with utility theory and the real-life behaviour of individuals' decision-making under risk \cite{Kahneman_1979}. The main elements that constitute this theory are \cite{Barberis_2013}: 
\begin{itemize}
\item reference dependence, which is a measure of utility to a reference point;
\item loss aversion, which captures the idea that people are more sensitive to losses;
\item diminishing sensitivity, which models losses and gains in different ways, i.e. the value function is convex in the region of losses and concave in the region of gains, respectively;
\item probability weighting, which account for the addition weights that individuals assign to probabilities of events.
\end{itemize} 
Through the pairwise comparison dynamics, we capture some elements of irrational or subjective behaviour that comes from the above elements of prospect theory. Prospect theory has been applied in a range of settings, including discrimination in obsessive-compulsive disorders \cite{George_2019}, dynamics pricing for shared mobility \cite{Guan_2019} and in the context of collective decision-making through pairwise comparison dynamics \cite{Stella_2018}. Finally, we include prospect theory in the formulation of the OPS with game-environment feedback via pairwise comparison dynamics through the use of a weighting function which gives weight zero to the probability of unfavourable events in a risk-seeking scenario.

\textit{Highlights of contributions}. The contribution of this paper is threefold. First, we propose a modelling framework where we capture the irrational behaviours originating in the context of behavioural economics through Kahneman and Tversky's prospect theory. In the proposed framework, we study the game-environment feedback on replicator dynamics\cite{Weitz_2016} and we extend it to a different kind of dynamics, i.e. pairwise comparison dynamics. Second, we carry out the stability analysis of the system resulting from these dynamics and discuss its differences with the original model. Lastly, we extend the work on the PD to account for the abstain strategy via the OPD. Motivated by the importance of the abstain strategy in real-life contexts and the irrational behaviours in many decision making settings, we investigate the system resulting from the application of pairwise comparison dynamics to the optional game. 

This paper is organised as follows. In Section~\ref{sec:PD}, we provide the problem formulation for the PD with game-environment feedback in the case of pairwise comparison dynamics. We conduct the stability analysis of the equilibria of the system and we discuss the differences and similarities between our results and the results in previous research works. In Section~\ref{sec:OPD}, we extend the proposed framework with game-environment feedback to the case of the OPD, where players can choose to abstain from playing the game as a third strategy. We study the stability of the two models and discuss the main differences between them and the added value of this work in relation to the rest of the literature on this topic.
Finally, in Section~\ref{sec:conclusion} we draw conclusions and discuss future directions of research.

\section{Feedback-Evolving Prisoner's Dilemma via Pairwise Comparison Dynamics}\label{sec:PD}
Evolutionary game dynamics with game-environment feedback have become popular in the past few years because of the induced dependency between the population dynamics and the surrounding environment. This leads to a system where the environment changes as a consequence of the frequency of the strategies and where the population dynamics are affected by the environment in return. Figure~\ref{fig:feedback} shows a diagram representation of the game-environment feedback, where the game dynamics and the environment form a closed-loop system.

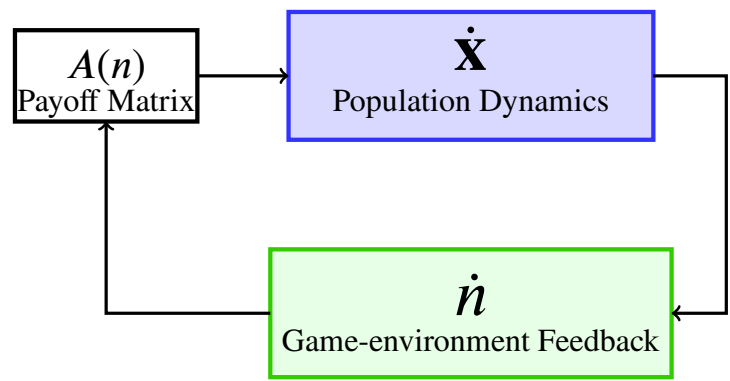
\begin{figure}[h!]
\centering
	\begin{tikzpicture}
	\begin{scope}[shift={(0,0)},scale=1.2]
		\draw [color=green!80!red, fill=green!10, ultra thick] (2.8,0) rectangle (7.2,1.4);
		\node at (5,0.9) {\Huge $\dot n$};
		\node at (5,0.4) {\large Game-environment Feedback};
		
		\draw [ultra thick] (0,2.8) rectangle (2,3.8);
		\node at (1,3.4) {\Large $A(n)$};
		\node at (1,3) {\large Payoff Matrix};
		
		\draw [color=blue!80!white, fill=blue!15, ultra thick] (3,2.6) rectangle (7,4);
		\node at (5,3.6) {\Huge $\dot{\mathbf x}$};
		\node at (5,3) {\large Population Dynamics};

		\draw [->,very thick]  (2.8,0.7) -- (1,0.7) -- (1,2.8);
		\draw [->,very thick]  (2,3.3) -- (3,3.3);
		\draw [->,very thick]  (7,3.3) -- (7.8,3.3) -- (7.8,0.7) -- (7.2,0.7);
	\end{scope}
	\end{tikzpicture}
	\caption{Diagram representation of the evolutionary dynamics for game-environment feedback. In general, the payoff matrix $A$ determines the frequencies of strategies, $x_i$ for all $i$, to derive the population dynamics. Through the game-environment feedback the frequencies of the strategies change the state of the environment, $n$, which, in turn, influences the values in the payoff matrix $A(n)$, with explicit dependence on the environmental resource $n$.}
\label{fig:feedback}
\end{figure}

\subsection{Evolutionary Game Model via Pairwise Comparison Dynamics}
In this section, we extend the initial feedback evolving game framework to pairwise comparison dynamics, which are a form of innovative dynamics. We compare our results to the framework obtained via replicator dynamics by Weitz \emph{et al.}\cite{Weitz_2016} in the context of the prisoner's dilemma (PD). The motivation for this extension is twofold: first, through pairwise comparison dynamics we can capture some aspects of irrational behaviour of individuals' decision-making under risk through prospect theory \cite{Barberis_2013, Kahneman_1979}; second, the results presented here share similarities with the seminal work on game-environment feedback that are worth investigating, especially in terms of equilibria and stability \cite{Weitz_2016}.

Let us consider the standard payoff matrix in the context of the PD where the two-player game consists of the two strategies cooperation (C) and defection (D):
\begin{equation}\label{eq:APD}
A = \left[ \begin{array}{cc}
R & S \\
T & P \end{array} \right],
\end{equation}
where $R$ is the reward for cooperating, $S$ is the sucker's payoff, $T$ is the payoff associated with the temptation to cheat and $P$ is the punishment for cheating. In the classical formulation of the game, $T > R > P > S$ so that mutual defection is the Nash equilibrium associated with the game. In the context of evolutionary game theory, the model associated with the above  symmetric two-player game via replicator dynamics can be written as \cite{Weitz_2016}
\begin{equation}\label{eq:model}
\dot x = x (1 - x) (r_1(x) - r_2(x)),
\end{equation}
where $x \equiv x_1$ because of the conservation of mass law for which $x_1 + x_2 = 1$, and where $r_1(x)$ is the fitness of player 1 and $r_2(x)$ is the fitness of player two, respectively. These are defined as:
\begin{eqnarray}
r_1(x) &=& \sum_j a_{1j} x_j = (Ax)_1 = Rx + S(1-x),\\
r_2(x) &=& \sum_j a_{2j} x_j = (Ax)_2 = Tx + P(1-x).
\end{eqnarray}
After substituting the above into equation~(\ref{eq:model}), we obtain:
\begin{equation}\label{eq:PDmodel2}
\dot x = x (1 - x) (Rx + S(1- x) - Tx - P(1 - x)) = -x (1 - x) (\delta_{TR} x + \delta_{PS} (1-x)),
\end{equation}
where $\delta_{TR} = T - R$ and $\delta_{PS} = P - S$.

\medskip
We can now formulate the analogous model for the PD via pairwise comparison dynamics. To this end, let us consider the following dynamics:
\begin{equation}\label{eq:PCD}
\dot x_i = \sum_j x_j \phi_{ji} - x_i \sum_j \phi_{ij},
\end{equation}
where $\phi_{ij} \ge 0$ is the rate at which a player $i$ changes his strategy to strategy $j$ and $x_i\phi_{ij}$ is the flux from strategy $i$ to strategy $j$. When $\phi_{ij}$ is given by the proportional rule $\phi_{ij} = [a_j - a_i]_+ = \sum_k [a_{jk} - a_{ik}]_+x_j$, and by recalling that the notation $[\cdot]_+$ denotes the positive part as before, the resulting dynamics are called pairwise comparison dynamics or pairwise difference dynamics \cite{Hofbauer_2011}. After a simple calculation, the model resulting from the pairwise comparison dynamics for the PD is the following:
\begin{equation}\label{eq:modelPCD}
\dot x = (1 - x) [r_1(x) - r_2(x)]_+ - x [r_2(x) - r_1(x)]_+,
\end{equation}
where $r_1(x)$ and $r_2(x)$ are the fitness of player 1 and the fitness of player 2, respectively, as defined before. If we substitute the payoffs from matrix~(\ref{eq:APD}), we get the following model:
\begin{equation}\label{eq:modelPCD2}
\dot x = - x (\delta_{TR} x + \delta_{PS} (1 - x)).
\end{equation}
\begin{lemma}
Consider system~(\ref{eq:modelPCD2}). This system has two fixed points: $x^* = 0$ which is asymptotically stable and $x^* = 1 - \delta_{PS}/\delta_{TR}$ which is unstable. 
\end{lemma}
\begin{proof}
The stability of these two equilibria can be identified by looking at the sign of the square, i.e. $x^* = 0$ is stable whereas $x^* =  \delta_{PS}/(\delta_{PS} - \delta_{TR})$ is unstable. Therefore, when $t$ is large, the population will converge to the point where all players choose to defect and no cooperators are present.
\end{proof}

\medskip
\textit{Example 1}. Consider the following payoff matrix
$$
A = \left[ \begin{array}{cc}
3 & 0 \\
5 & 1 \end{array} \right].
$$
We can calculate the fitness of player 1 and the fitness of player 2 as in the following:
$$
r_1(x) = 3x_1,
$$
$$
r_2(x) = 5x_1 + x_2,
$$
and the resulting model via replicator dynamics is
$$
\dot x = -x (1 - x) (1 + x),
$$
while the model resulting from the pairwise difference dynamics is
$$
\dot x = -x (1 + x).
$$
In the above, regardless of the dynamics used, the only stable equilibrium point in the domain $[0,\, 1]$ is $x^* = 0$, as it can be deduced from the sign as before. Within the same domain, the model obtained via replicator dynamics presents another equilibrium point $x^* = 1$, which is unstable. The equilibrium point $x^* = -1$, albeit common to both models, is not in the given domain and therefore is not considered. The dynamics for this example are depicted in Fig.~\ref{fig:examples12} (top). \hfill $\square$

\begin{figure}[h!]
\centering
	\includegraphics[width=0.75\linewidth]{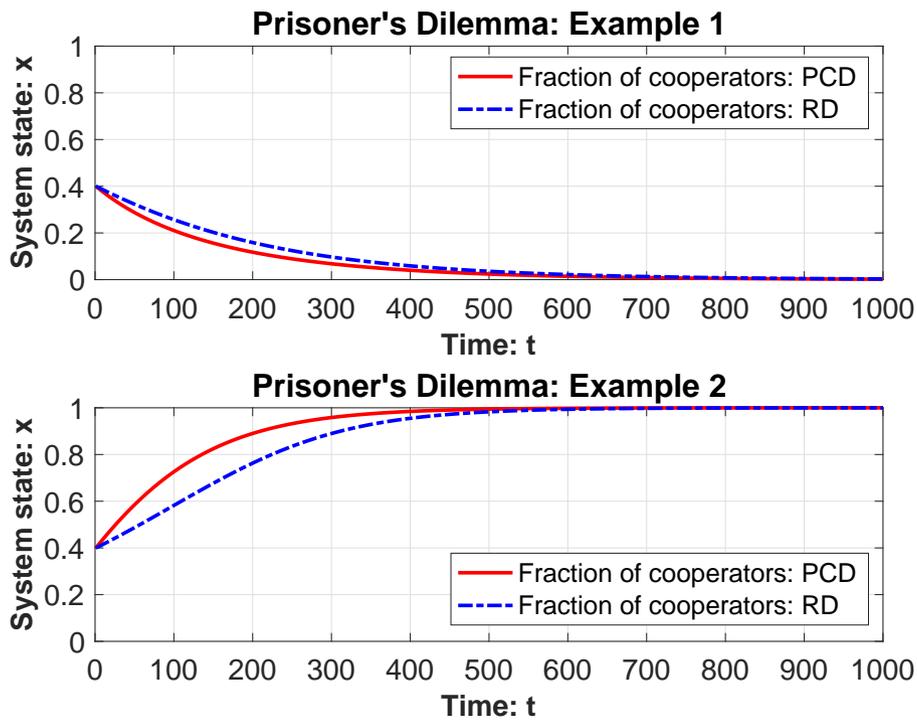}
	\caption{Time evolution of the model for the PD via replicator dynamics (RD) and pairwise comparison dynamics (PCD). The first plot (top) shows the evolution of the cooperation strategy for \textit{Example 1}. The second plot (bottom) shows the evolution of the cooperation strategy for \textit{Example 2}.}
\label{fig:examples12}
\end{figure}

\medskip
\textit{Example 2}. In this example, we consider a situation where cooperation is favoured. Let the payoff matrix be
$$
A = \left[ \begin{array}{cc}
5 & 1 \\
3 & 0 \end{array} \right].
$$
The fitness of player 1 and the fitness of player 2 are calculated as
$$
r_1(x) = 5x_1 + x_2,
$$
$$
r_2(x) = 3x_1,
$$
and the resulting model via replicator dynamics is
$$
\dot x = x (1 - x) (1 + x),
$$
while the model resulting from the pairwise difference dynamics is
$$
\dot x = (1 - x) (1 + x).
$$
As in the previous example, the model via replicator dynamics has two equilibria in the domain $[0,\, 1]$, namely $x^* = 0$ and $x^* = 1$, whereas the model with pairwise comparison dynamics has only one equilibrium $x^* = 1$. In both cases, the only stable equilibrium point is $x^* = 1$, meaning that cooperation is the stable strategy. The dynamics for this example are depicted in Fig.~\ref{fig:examples12} (bottom). \hfill $\square$

\subsection{Game-Environment Feedback}
We can now extend the above results in the context of feedback-evolving games, namely when the payoff matrix depends on an environmental resource $n$, in line with the work by Weitz \emph{et al.} \cite{Weitz_2016}. We define our game-environment feedback dynamics as
\begin{equation}\label{eq:GED2}
\begin{array}{lll}
\epsilon \dot x = (1 - x) [r_1(x,A(n)) - r_2(x,(A(n))]_+ - x [r_2(x,A(n)) - r_1(x,A(n))]_+, \\
\dot n = n (1 - n) [(1 + \lambda) x -1].
\end{array}
\end{equation}
The main difference between the above model and the one in the previous section is the environmental dependency of the payoff matrix $A(n)$. For the purpose of studying this model, we consider the following environment-dependent matrix:
\begin{equation}\label{eq:APDn}
A(n) = (1 - n)\left[ \begin{array}{cc}
T & P \\
R & S \end{array} \right] +
n\left[ \begin{array}{cc}
R & S \\
T & P \end{array} \right] = 
\left[ \begin{array}{cc}
T - n\delta_{TR} & P - n\delta_{PS} \\
R + n\delta_{TR} & S + n\delta_{PS} \end{array} \right],
\end{equation}
where $\delta_{TR} = T - R$, $\delta_{PS} = P - S$ and we assume $T > R > P > S$ as before. The resulting pairwise comparison dynamics (pairwise difference) for the PD with game-environment feedback is therefore:
\begin{equation}\label{eq:GED3}
\begin{array}{lll}
\epsilon \dot x = (1 - x) [\delta_{PS} + (\delta_{TR} - \delta_{PS})x] (1-2n), \\
\dot n = n (1 - n) [(1 + \lambda) x -1].
\end{array}
\end{equation}

\medskip
\begin{theorem}
System~(\ref{eq:APDn}) has five fixed points, namely: $(x^*, n^*) = (1, 0)$, $(x^*, n^*) = (1, 1)$, $(x^*, n^*) = (\delta_{PS}/(\delta_{PS} - \delta_{TR}), 0)$, $(x^*, n^*) = (\delta_{PS}/(\delta_{PS} - \delta_{TR}), 1)$ and $(x^*, n^*) = (1/(\lambda+1), 1/2)$. The first four fixed points are unstable, and the last one, namely $(x^*, n^*) = (1/(\lambda+1), 1/2)$, is a neutrally stable center.
\end{theorem}
\begin{proof}
System~(\ref{eq:APDn}) has five fixed points, namely: $(x^*, n^*) = (1, 0)$, $(x^*, n^*) = (1, 1)$, $(x^*, n^*) = (\delta_{PS}/(\delta_{PS} - \delta_{TR}), 0)$, $(x^*, n^*) = (\delta_{PS}/(\delta_{PS} - \delta_{TR}), 1)$ and $(x^*, n^*) = (1/(\lambda+1), 1/2)$. The first two correspond to the scenario of cooperators in a depleted environment and cooperators in a replete environment. The third and fourth equilibria depend on the payoffs, whereas the last equilibrium point is the only interior fixed point representing a population of cooperators and defectors in a half-depleted environment.
Without loss of generality, we set $\epsilon = 1$ and carry out the stability analysis of the above equilibria. The Jacobian of system~(\ref{eq:APDn}) is the following:
\begin{equation}\label{eq:J2}
J(x,n) = \left[ \begin{array}{cc}
(1 - 2n)(\delta_{TR} - 2\delta_{PS}) + x(4n - 2)(\delta_{TR} - \delta_{PS}) & \quad -2\delta_{PS} -2x\delta_{TR} 4x\delta_{PS} +2x^2(\delta_{TR} - \delta_{PS}) \\
n (\lambda + 1) - n^2(\lambda +1) & \quad -1 + 2n + x(\lambda +1) - 2nx (\lambda +1) \end{array} \right].
\end{equation}
The Jacobian linearised at the first four fixed points is:
\begin{equation}\nonumber
\begin{array}{ll}
J(1,0) = \left[ \begin{array}{cc}
-\delta_{TR} & 0 \\
0 & \lambda \end{array} \right], \quad
& J\Big(\frac{\delta_{PS}}{\delta_{PS} - \delta_{TR}},0\Big) = \left[ \begin{array}{cc}
\delta_{TR} & 0 \\
0 & \frac{\lambda \delta_{PS} + \delta_{TR}}{\delta_{PS} - \delta_{TR}} \end{array} \right], \\
J(1,1) = \left[ \begin{array}{cc}
\delta_{TR} & 0 \\
0 & -\lambda \end{array} \right], \quad
\quad & J\Big(\frac{\delta_{PS}}{\delta_{PS} - \delta_{TR}},1\Big) = \left[ \begin{array}{cc}
-\delta_{TR} & 0 \\
0 & \frac{\lambda \delta_{PS} - \delta_{TR}}{\delta_{PS} - \delta_{TR}} \end{array} \right].
\end{array}
\end{equation}
Each Jacobian has at least one positive eigenvalue and is therefore locally unstable. We can now turn our attention to the interior fixed point and calculate its Jacobian as in the following:
\begin{equation}\nonumber
J\Big(\frac{1}{\lambda + 1},\frac{1}{2}\Big) = \left[ \begin{array}{cc}
0 & \quad - 2 \delta_{PS} (\lambda + 1)^2  + 2 (\lambda + 1) (2 \delta_{PS} - \delta_{TR}) - 2(\delta_{PS} - \delta_{TR}) \\
\frac{1}{4}(\lambda + 1) & \quad 0 \end{array} \right],
\end{equation}
whose eigenvalues have no real parts and thus the interior fixed point is a neutrally stable center.
\end{proof}

\begin{remark}
By comparing the above formulation of the PD with the one in Weitz \emph{et al.} \cite{Weitz_2016}, we obtain a model of degree 2 whose fixed points share the same stability properties with the one resulting from the replicator dynamics. Two interesting differences emerge from comparing the two models. The first one is that our model does not directly include the fixed points corresponding to all defectors in a depleted/replete environment, unless $P = S$. The second one can be found by looking at the degree of our system, which is one lower than the one by Weitz \emph{et al.} This can be explained by looking at the rate of changing strategy in the pairwise difference comparison: the value obtained by the proportional rule $\phi_{ij} = [a_j - a_i]_+$ considers the probability of a payoff increase only, thus assigning a different weighting factor to different events. This aspect is one of the core elements in prospect theory and is the motivation that underpins our approach \cite{Kahneman_1979}.
\end{remark}

\medskip
\textit{Example 3}. Given the payoff matrix in~(\ref{eq:GED2}), for the sake of clarity we now consider system~(\ref{eq:APDn}) with the following parameters:
\begin{equation}\nonumber
\begin{array}{lll}
&\epsilon = 0.1, \quad \lambda = 2,\\
R = 3, \quad &S = 0, \quad T = 5, \quad P = 1.
\end{array}
\end{equation}
The model resulting from the pairwise comparison dynamics is
\begin{equation}\nonumber
\begin{array}{lll}
0.1 \dot x = (1 - x) [1 + x](1 - 2n),\\
n = n(1-n)[3x - 1].
\end{array}
\end{equation}
In accordance with the results, the above model has five fixed points: $(x^*, n^*) = (1, 0)$, $(x^*, n^*) = (1, 1)$, $(x^*, n^*) = (-1, 0)$, $(x^*, n^*) = (-1, 1)$ and $(x^*, n^*) = (1/3, 1/2)$. The third and the fourth are outside the domain $[0,\, 1]$, and thus are not considered. The first two fixed points are unstable, whereas the interior fixed point is a neutrally stable center, and therefore we have an oscillatory behaviour. The time evolution and phase plane dynamics for this example are depicted in Fig.~\ref{fig:example3}. \hfill $\square$

\begin{figure}[t]
\centering
	\includegraphics[width=\linewidth]{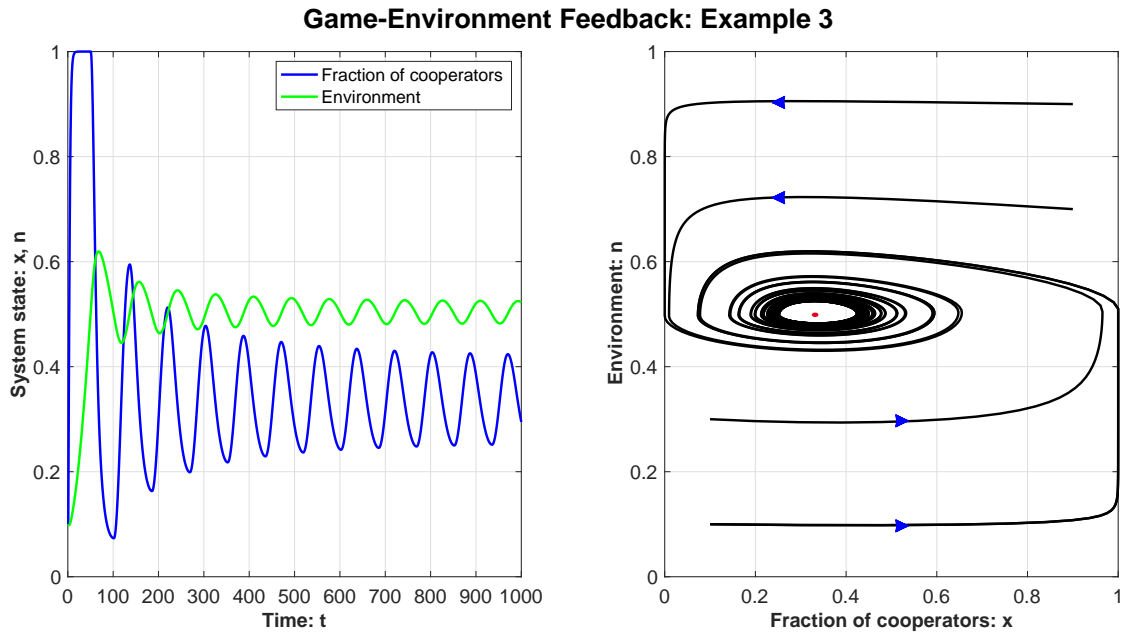}
	\caption{Feedback-evolving game model resulting from the pairwise comparison dynamics for the PD (\textit{Example 3}). The first plot (left) shows the evolution of the fraction of cooperators and the environment over time. The second plot (right) shows the phase plane dynamics in the $x$-$n$ plane for system ~(\ref{eq:APDn}): the blue arrows indicate the direction of dynamics while the red dot denotes the interior fixed point at (1/3, 1/2). Finally, the different trajectories correspond to the following initial conditions: $(0.9, 0.9)$, $(0.9, 0.7)$, $(0.1, 0.3)$, $(0.1, 0.1)$.}
\label{fig:example3}
\end{figure}

\section{Optional Prisoner's Dilemma with game-environment Feedback}\label{sec:OPD}
A crucial aspect that is not captured by the standard PD is the possibility that players can choose neither option. The optional prisoner's dilemma (OPD) game is an extension of the PD where players can choose to abstain, and therefore neither cooperate nor defect \cite{Cardinot_2016}. This translates to a game with three possible strategies: cooperate, defect and abstain. The implications of this additional strategy are crucial in many applications, whose most notable examples are opinion dynamics during an election and duopolistic competition in marketing \cite{Batali_1995, Stella_2019}.

To describe the evolution of the strategy profile in the population under game-environment feedback, let the environment-dependent payoff matrix be:
\begin{equation}\label{eq:AOPDn}
A(n) = (1 - n)\left[ \begin{array}{ccc}
T & P & L \\
R & S & L \\
L & L & L \end{array} \right] +
n\left[ \begin{array}{ccc}
R & S & L \\
T & P & L \\
L & L & L \end{array} \right] = 
\left[ \begin{array}{ccc}
T - n\delta_{TR} & P - n\delta_{PS} & L \\
R + n\delta_{TR} & S + n\delta_{PS} & L \\
L & L & L \end{array} \right],
\end{equation}
where the parameters have the usual meaning, and $L$ is the loner's payoff. This payoff is given to both players if at least one chooses to abstain. As it is common in the literature, we assume $T > R > L > P > S$. 

In the following, we discuss the feedback-evolving OPD game model via replicator dynamics and pairwise comparison dynamics. First, we provide a general formulation of the OPD via replicator dynamics as in the following:
\begin{equation}\label{eq:OPDrd}
\begin{array}{lll}
\epsilon \dot x_i = x_i(1 - x_i) [r_i(x, A(n))] - x_i x_j [r_j(x,A(n))] - x_i L (1 - x_i - x_j), \\
\dot n = n (1 - n) [(1 + \lambda) x_1 -1],
\end{array}
\end{equation}
where $i = 1, 2$ and $j = 2, 1$. The specific model resulting from the payoff matrix~(\ref{eq:AOPDn}) where $x_3 = 1 - x_1 - x_2$ is the following:
\begin{equation}\label{eq:OPDrd2}
\begin{array}{lll}
\epsilon \dot x_1 = x_1(1 - x_1) [r_1(x, A(n))] - x_1 x_2 [r_2(x,A(n))] - x_1 L (1 - x_1 - x_2), \\
\epsilon \dot x_2 = x_2(1 - x_2) [r_2(x, A(n))] - x_1 x_2 [r_1(x,A(n))] - x_2 L (1 - x_1 - x_2), \\
\dot n = n (1 - n) [(1 + \lambda) x_1 -1].
\end{array}
\end{equation}
In the above model, the evolution of the environment is subject to the influence of the players choosing to cooperate, similarly to previous works \cite{Weitz_2016, Tilman_2020}. To derive the corresponding model via pairwise comparison dynamics, recall that the rate $\phi_{ij}$ is given by the proportional rule $\phi_{ij} = [a_j - a_i]_+ = \sum_k [a_{jk} - a_{ik}]_+x_j$. Therefore we have:
\begin{equation}\label{eq:OPDpcd}
\begin{array}{lll}
\epsilon \dot x_i = x_j [\phi_{ji}]_+ + (1 - x_i - x_j) [\phi_{3i}]_+ - x_i ( [\phi_{ij}]_+ + [\phi_{i3}]_+ ), \\
\dot n = n (1 - n) [(1 + \lambda) x_1 -1],
\end{array}
\end{equation}
where $i = 1, 2$ and $j = 2, 1$ as before. By substituting the values from the payoff matrix~(\ref{eq:AOPDn}), the resulting model is the following:
\begin{equation}\label{eq:OPDpcd2}
\begin{array}{lll}
\epsilon \dot x_1 = (x_2 - 3n x_1) (\delta_{TR} x_1 + \delta_{PS} x_2) + \delta_{TL} x_1 (1 - x_1 - x_2) + \delta_{PL} x_1 x_2, \\
\epsilon \dot x_2 = (x_1 - 2n x_2) (\delta_{TR} x_1 + \delta_{PS} x_2) + (1 - x_1 - x_2) ((\delta_{RL} + \delta_{TR}n) x_1 + \delta_{PS} n x_2) + \delta_{TL} x_1 x_2, \\
\dot n = n (1 - n) [(1 + \lambda) x_1 -1],
\end{array}
\end{equation}
where $\delta_{TL} = T - L$, $\delta_{PL} = P - L$ and $\delta_{RL} = R - L$, while the other parameters have the usual meaning.

\begin{center}
\begin{table}[t]%
	\centering
	\caption{List of all 11 fixed points for system~(\ref{eq:OPDrd2}).}%
	\label{t:eqRD}
	\begin{tabular*}{500pt}{@{\extracolsep\fill}lccc@{\extracolsep\fill}}
		\toprule
		\# & $x_1$ & $x_2$  & $n$ \\
		\midrule
		1 & 0 & 0 & 0   \\
		2 & 1 & 0 & 0   \\
		3 & 0 & 1 & 0   \\
		4 & 0 & 0 & 1   \\
		5 & 1 & 0 & 1   \\
		6 & 0 & 1 & 1   \\
		7 & $\delta_{PS}/(\delta_{PS}-\delta_{TR})$ & $-\delta_{TR}/(\delta_{PS}-\delta_{TR})$ & 0   \\
		8 & $\delta_{PS}/(\delta_{PS}-\delta_{TR})$ & $-\delta_{TR}/(\delta_{PS}-\delta_{TR})$ & 1   \\
		9 & $1/(1+\lambda)$ & $\lambda/(1+\lambda)$ & $1/2$   \\
		10 & $1/(1+\lambda)$ & $-(R - 2L +T)/((P-2L+S)(1+\lambda))$ & $1/2$   \\
		11 & $1/(1+\lambda)$ & 0 & $(L-T)/(R-T)$   \\
		\bottomrule
	\end{tabular*}
\end{table}
\end{center}

\medskip
\textbf{Results} -- \textit{Stability Analysis}.
System~(\ref{eq:OPDrd2}) has a total of eleven fixed points, as reported in Table~\ref{t:eqRD}. The first one $(0,0,0)$, namely all players abstain, is marginally stable as it can be seen by the Jacobian linearised about it:
\begin{equation}\nonumber
J(0,0,0) = \left[ \begin{array}{ccc}
0 & 0 & 0 \\
0 & 0 & 0 \\
0 & 0 & -1 \end{array} \right],
\end{equation}
Depending on the initial conditions, the trajectories can converge to this point or be attracted to the interior fixed points 9 and 10 from Table~\ref{t:eqRD}, similarly to the PD with game-environment feedback. On a preliminary analysis, this second behaviour can be seen when no players choose to abstain at the start of the game, namely $x_1(0) + x_2(0) = 1$.

On the other hand, system~(\ref{eq:OPDpcd2}) has ten fixed points, but the analytical expression for each of them is too large to fit in a table, albeit similar to the one for the replicator dynamics. The important element in this case is that, regardless of the initial conditions, all trajectories converge to the interior fixed point at which the number of cooperators is equal to $1/(1+\lambda)$, which turns out to be asymptotically stable.

\begin{remark}
It is worth noting that by unilaterally choosing to assign weight zero to the probability of unfavourable events and thus by considering only the probability of a payoff increase via pairwise comparison dynamics, the population converges to a mixed strategy equilibrium, where cooperation is achieved by one part of the population. This formulation is in accordance to the main elements that constitute the prospect theory \cite{Kahneman_1979}. Finally, by considering a situation where the degrading effect is stronger, namely $\lambda < 1$, the population will converge to a strategy profile where the dominant strategy is cooperation.
\end{remark}


\medskip
\textit{Example 4}. Let the payoff matrix be (\ref{eq:AOPDn}) with the following parameters:
\begin{equation}\nonumber
\begin{array}{lll}
&\epsilon = 0.5, \quad \lambda = 2,\\
R = 3, \quad S = 0, & \quad L = 2, \quad T = 5, \quad P = 1.
\end{array}
\end{equation}
The models resulting from the replicator dynamics and pairwise difference dynamics share many fixed points as before, but the stability behaviour of these points changes in relation to the dynamics used and the initial condition. For the replicator dynamics, we observe oscillations when $x_1 + x_2 = 1$, namely when no players abstain. The trajectories have an oscillating behaviour about the internal fixed point, which can be calculated as (0.33, 0.44, 0.5). The only stable fixed point is $(0,0,0)$, namely when all players abstain. In the case of pairwise comparison dynamics, we observe a completely different behaviour: all trajectories converge to the internal fixed point, i.e. (0.33, 0.37, 0.51), regardless of the initial condition. This internal point is asymptotically stable. These behaviours can be seen in Figg.~\ref{fig:OPD1}-\ref{fig:OPD2}. \hfill $\square$

\begin{figure}[t]
\centering
	\includegraphics[width=\linewidth,height=0.325\linewidth]{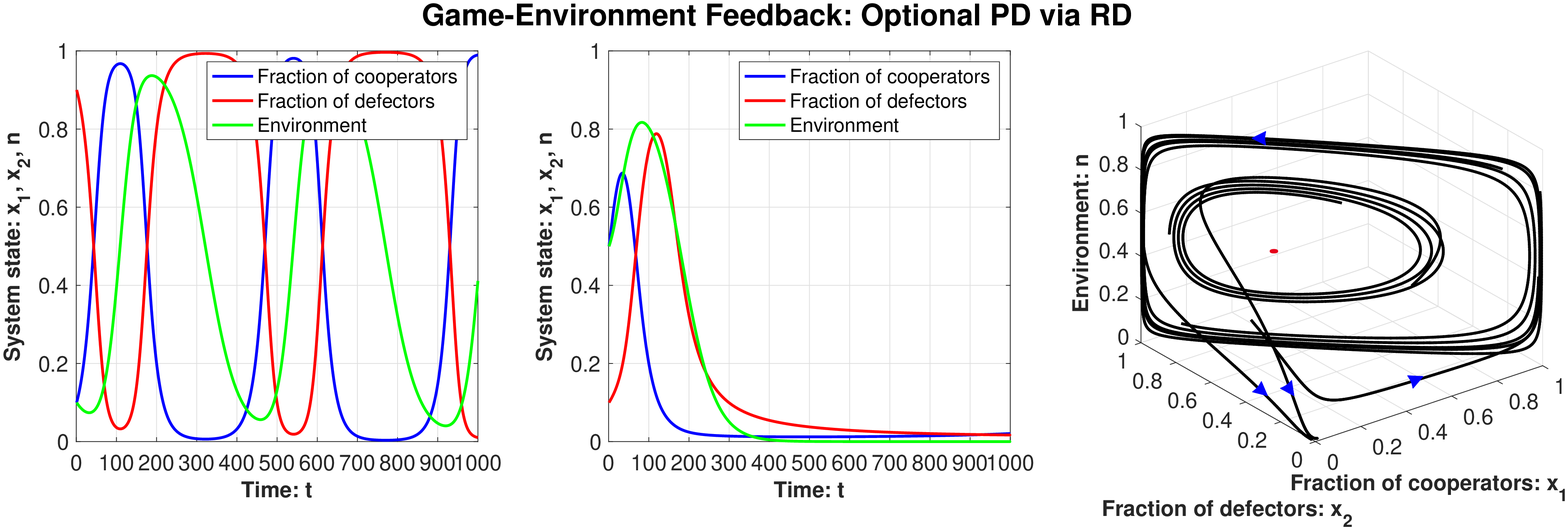}
	\caption{Feedback-evolving game model resulting from the replicator dynamics for the OPD (\textit{Example 4}). The first two plots (left and centre) show the evolution of the fraction of cooperators and defectors and the environment over time for initial conditions $(0.9,0.1,0.1)$ and $(0.1,0.9,0.9)$, respectively. The last plot (right) shows the phase plane dynamics in the $x_1$-$x_2$-$n$ space: the blue arrows indicate the direction of the dynamics while the red dot denotes the interior fixed point at (0.33, 0.44, 0.5). For different initial conditions, the trajectories are either trapped in a limit cycle or converge to $(0,0,0)$.}
\label{fig:OPD1}
\end{figure}

\begin{figure}[h!]
\centering
	\includegraphics[width=\linewidth,height=0.325\linewidth]{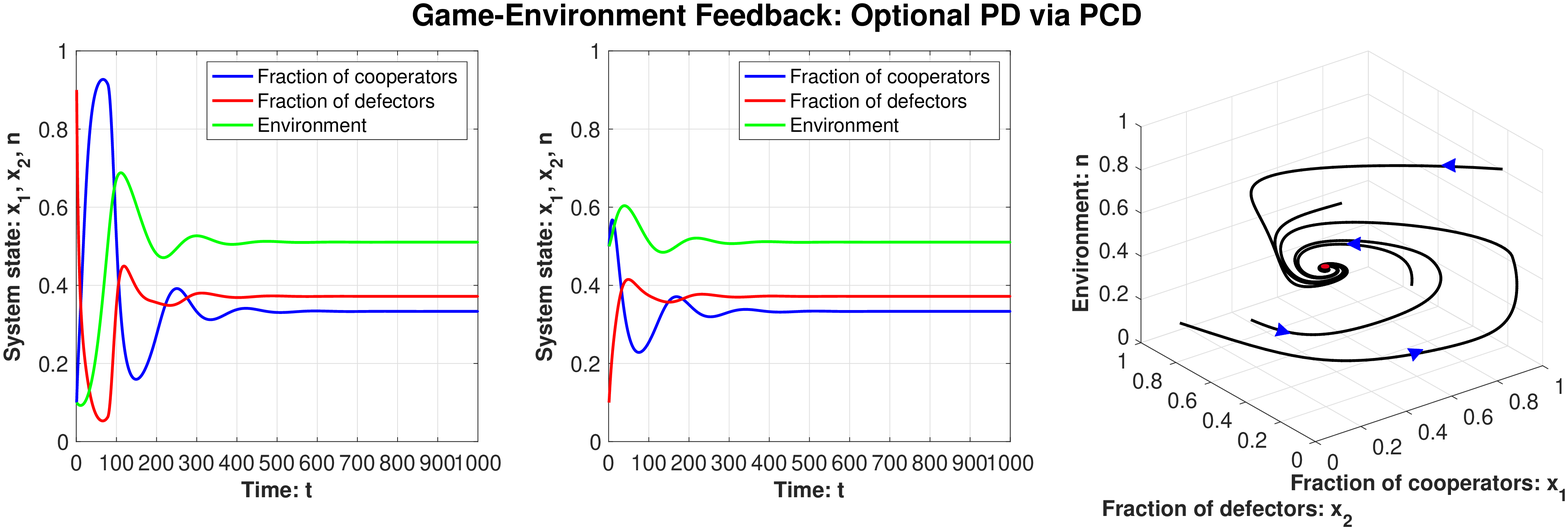}
	\caption{Feedback-evolving game model resulting from the pairwise comparison dynamics for the OPD (\textit{Example 4}). The first two plots (left and centre) show the evolution of the fraction of cooperators and defectors and the environment over time for initial conditions $(0.9,0.1,0.1)$ and $(0.1,0.9,0.9)$, respectively. The last plot (right) shows the phase plane dynamics in the $x_1$-$x_2$-$n$ space: the blue arrows indicate the direction of the dynamics while the red dot denotes the interior fixed point at (0.33, 0.37, 0.51). Regardless of the initial conditions, all trajectories converge to the interior fixed point.}
\label{fig:OPD2}
\end{figure}                    

\section{Conclusions}\label{sec:conclusion}
Motivated by the emerging interest towards a game theoretic modelling framework that can capture player's irrational behaviours and the possibility to abstain in a feedback-evolving game context, in this paper we have proposed a unified framework to capture these three main aspects. The resulting feedback-evolving optional prisoner's dilemma under irrational behaviour is obtained via pairwise comparison dynamics. We have discussed the differences between our model and the prior works on the topic of game-environment feedback within evolutionary game theory. Driven by the need for capturing optional games, we have extended our analysis to the optional prisoner's dilemma by deriving the model corresponding to the replicator dynamics and the pairwise comparison dynamics. Interestingly, the behaviour of these two systems under game-environment feedback is notably different, despite the presence of oscillations in both models. The latter model gives rise to an interesting setting for partial cooperation, which provides an explanation for the reasons underpinning the presence of cooperation in many biological systems. Future work includes: i) the extension of this model to a structured setting where a network topology is used to capture the interactions in the population, ii) the study of different time scales for the population dynamics and the evolution of the environment, and iii) the impact of time-varying parameters on the population dynamics.

\subsection*{Author contributions}

L.S. and D.B. designed the research, L.S. performed the research and wrote the paper.

\subsection*{Financial disclosure}

None reported.

\subsection*{Conflict of interest}

The authors declare no potential conflict of interests.

\end{document}